\documentclass{article}
\usepackage{eurosym}
\usepackage{amsmath}
\usepackage{amsfonts}
\usepackage{tabularx}
\usepackage[table]{xcolor}

\newtheorem{theorem}{Theorem}

\newtheorem{corollary}[theorem]{Corollary}

\newtheorem{example}[theorem]{Example}

\newtheorem{lemma}[theorem]{Lemma}

\newtheorem{remark}[theorem]{Remark}

\newenvironment{proof}[1][Proof]{\noindent\textbf{#1.} }{\ \rule{0.5em}{0.5em}}
\newcommand{\ZZ}{{\mathbb{Z}_{2}\mathbb{Z}_{2}[u]}}
\newcommand{\C}{{\mathcal{C}}}

\newcommand{\Z}{{\mathbb{Z}}}

\begin{document}

\title{Self-Dual Cyclic and Quantum Codes Over $\mathbb{Z}%
_{2}^{\alpha}\times(\mathbb{Z}_{2}+u{\mathbb{Z}}_{2})^{\beta}$}
\author{Ismail Aydogdu$^{a}$\thanks{{\footnotesize iaydogdu@yildiz.edu.tr
(Ismail Aydogdu)}} and Taher Abualrub $^{b}$\thanks{{\footnotesize %
abualrub@aus.edu (Taher Abualrub )}} \\
$^{a}${\footnotesize Department of Mathematics, Yildiz Technical University}%
\\
$^{b}${\footnotesize Department of Mathematics and Statistics, American
University of Sharjah}}
\maketitle

\begin{abstract}
In this paper we introduce self-dual cyclic and quantum codes over $\mathbb{Z%
}_{2}^{\alpha }\times (\mathbb{Z}_{2}+u{\mathbb{Z}}_{2})^{\beta }$. We
determine the conditions for any $\mathbb{Z}_{2}\mathbb{Z}_{2}[u]$-cyclic
code to be self-dual, that is, ${\mathcal{C}}={\mathcal{C}}^{\perp }$. Since
the binary image of a self-orthogonal ${\mathbb{Z}_{2}\mathbb{Z}_{2}[u]}$%
-linear code is also a self-orthogonal binary linear code, we introduce
quantum codes over $\mathbb{Z}_{2}^{\alpha }\times \left( \mathbb{Z}_{2}+u%
\mathbb{Z}_{2}\right) ^{\beta }$. Finally, we present some examples of
self-dual cyclic and quantum codes that have good parameters.
\end{abstract}

\begin{quotation}
\end{quotation}

\section{Introduction}

Recently, inspired by the $\mathbb{Z}_2\mathbb{Z}_4$-additive codes
(introduced in \cite{3}), $\mathbb{Z}_2\mathbb{Z}_{2}[u]$-linear codes have
been introduced in \cite{2}. Though these code families are similar to each
other, $\mathbb{Z}_{2}\mathbb{Z}_{2}[u]$-linear codes have some advantages
compared to $\mathbb{Z}_{2}\mathbb{Z}_{4}$-additive codes. For example, the
Gray image of any $\mathbb{Z}_{2}\mathbb{Z}_{2}[u]$-linear code will always
be a linear binary code, but this property does not hold for $\mathbb{Z}_2%
\mathbb{Z}_4$-additive codes. Aydogdu et al. have studied cyclic and
constacyclic codes over this new structure in \cite{consta}.

Self-dual codes have an important place in coding theory because many of the
best codes known are of this type and also they have a very rich algebraic
property. Recently, there has been much interest in studying self-dual codes
over finite fields and rings. However, there do not exist much work for
self-dual cyclic codes. Binary self-dual cyclic codes were discussed by
Sloane and Thompson \cite{slone} and also self-dual cyclic codes over $%
F_{2}+uF_{2}$ were investigated in \cite{5}. Furthermore, Aydogdu and Kaya
have studied self-dual linear codes over ${\mathbb{Z}_{2}\mathbb{Z}_{2}[u]}$
more recently in \cite{ism}.

Quantum error-correcting codes play a significant role in quantum
computation and quantum information. One of the important classes of quantum
codes is the class of stabilizer codes. The CSS construction deals with the
self-orthogonal codes over the ring $R$ with $q$ elements. This construction
was introduced in 1996 by Calderbank and Shor \cite{shor} and Steane \cite%
{steane}. It provides the most direct link to classical coding theory. These
codes are interesting because they illustrate how classes of classical error
correcting codes can sometimes be turned into classes of quantum codes. They
are also useful for one of then simpler proofs of security of the BB84
quantum key-distribution protocol \cite{shor}. The work done by Calderbank
and Shor spurred many researches to study quantum error correcting codes
over different algebraic structures. The reader may see some of them in \cite%
{6,bieb,ozen,must,must1}.

In this paper, we are interested to study in self-dual cyclic codes over $\mathbb{Z}_{2}^{\alpha }\times R^{\beta }$ where $R=\Z_{2}+u\Z_{2}=\left\{ 0,1,u,u+1\right\} $ is the ring with four elements and $u^{2}=0.$ Then we
will apply these codes in the subject of quantum codes. The paper is
organized in the following way: Section 2 provides a preliminary about the
subject and the structure of cyclic codes and their duals over $\mathbb{Z}_{2}^{\alpha }\times R^{\beta }$. Section 3 gives the structure of self-dual $\ZZ$-cyclic codes. We will also provide some examples of
self-dual $\ZZ$-cyclic codes. In section 4, we study
quantum codes over $\mathbb{Z}_{2}^{\alpha }\times R^{\beta }$. We will also illustrate application of our theory by providing examples of optimal quantum codes. Section 5 gives a conclusion of the paper.

\section{Preliminary}

Let $R=\mathbb{Z}_{2}+u\mathbb{Z}_{2}=\left\{ 0,1,u,1+u\right\} $ be the
four-element ring with $u^{2}=0.$ A non-empty subset $\mathcal{C}$ of $%
\mathbb{Z}_{2}^{\alpha}\times R^{\beta}$ is called a $\mathbb{Z}_{2}\mathbb{Z%
}_{2}[u]$-linear code if it is an $R$-submodule of $\mathbb{Z}%
_{2}^{\alpha}\times R^{\beta}$. We know that the ring $R$ is isomorphic to $%
\mathbb{Z}_{2}^{2}$ as an additive group. Hence, for some positive integers $%
k_{0}$, $k_{2}$ and $k_{1}$, any $\mathbb{Z}_{2}\mathbb{Z}_{2}[u]$-linear
code $\mathcal{C}$ is isomorphic to an abelian group of the form $\mathbb{Z}%
_{2}^{k_{0}+k_{2}}\times \mathbb{Z}_{2}^{2k_{1}}$. We say such a code ${%
\mathcal{C}}$ is of type $\left(\alpha,\beta;k_{0};k_{1},k_{2}\right)$.
Furthermore the inner product of any two elements $V=\left( a_{0}a_{1}\ldots
a_{\alpha -1},b_{0}b_{1}\ldots b_{\beta -1}\right)$, $W=\left(
d_{0}d_{1}\ldots d_{\alpha-1},e_{0}e_{1}\ldots e_{\beta -1}\right) \in
\mathbb{Z}_{2}^{\alpha }\times R^{\beta }$ is defined by
\begin{equation*}
\left\langle V,W\right\rangle =\left( u\sum_{i=0}^{\alpha
-1}a_{i}d_{i}+\sum_{j=0}^{\beta -1}b_{j}e_{j}\right) \in \mathbb{Z}_{2}+u%
\mathbb{Z}_{2}.
\end{equation*}
The dual code ${\mathcal{C}}^{\perp}$ of ${\mathcal{C}}$ according to this
inner product is defined as follows.
\begin{equation*}
{\mathcal{C}}^{\perp}=\{W\in \mathbb{Z}_{2}^{\alpha}\times
R^{\beta}|~\langle V,W\rangle =0,\text{~for all}~V\in {\mathcal{C}} \}
\end{equation*}
Most of the concepts about the structure of ${\mathbb{Z}_{2}\mathbb{Z}_{2}[u]%
}$-linear codes have been given in \cite{2} with details.

An $R$-submodule $\mathcal{C}$ of $\mathbb{Z}_{2}^{\alpha }\times R^{\beta }$
is called a ${\mathbb{Z}_{2}\mathbb{Z}_{2}[u]}$-cyclic code if for any
codeword $v=\left( a_{0},a_{1},\ldots ,a_{\alpha -1},b_{0},b_{1},\ldots
,b_{\beta -1}\right) \in \mathcal{C},$ its cyclic shift
\begin{equation*}
T(v)=\left( a_{\alpha -1},a_{0},\ldots ,a_{\alpha -2},b_{\beta
-1},b_{0},\ldots ,b_{\beta -2}\right)
\end{equation*}%
is also in $\mathcal{C}$.

There is a one-to-one correspondence between the elements in $\mathbb{Z}%
_{2}^{\alpha }\times R^{\beta }$ and the elements in $R_{\alpha ,\beta }=%
\mathbb{Z}_{2}[x]/\left\langle x^{\alpha }-1\right\rangle \times
R[x]/\left\langle x^{\beta }-1\right\rangle .$ That is, any codeword $%
v=\left( a_{0}a_{1}\ldots a_{\alpha -1},b_{0}b_{1}\ldots b_{\beta -1}\right)
\in \mathbb{Z}_{2}^{\alpha }\times R^{\beta }$ can be identified as
\begin{eqnarray*}
v(x) &=&\left( a_{0}+a_{1}x+\ldots +a_{\alpha -1}x^{\alpha
-1},b_{0}+b_{1}x+\ldots +b_{\beta -1}x^{\beta -1}\right)  \\
&=&\left( a(x),b(x)\right) .
\end{eqnarray*}

\begin{theorem}
\cite{consta} \label{Main}Let $\mathcal{C}$ be a ${\mathbb{Z}_{2}\mathbb{Z}%
_{2}[u]}$-cyclic code in $R_{\alpha,\beta }.$ Then $\mathcal{C}$ is
identified uniquely as $\mathcal{C}=\langle \left( f(x),0\right) ,\left(
l(x),g(x)+ua(x)\right) \rangle ,$ where $f(x)|\left( x^{\alpha}-1\right)(%
\bmod~2),$ and $a(x)|g(x)|\left( x^{\beta}-1\right) \left(\bmod~2\right),$
and $l(x)$ is a binary polynomial satisfying $\deg (l(x))<\deg (f(x)),$ $%
f(x)|\left( \dfrac{x^{\beta}-1}{a(x)}\right) l(x)\left(\bmod~2 \right)$ and $%
f(x)\neq\left( \dfrac{x^{\beta}-1}{a(x)}\right)l(x)\left(\bmod~2\right)$. If
${\mathcal{C}}$ is generated by only $\left( l(x),g(x)+ua(x)\right)$ then $%
\left(x^{\alpha}-1\right)|\left(\dfrac{x^{\beta}-1}{a(x)}\right) l(x)\left(%
\bmod~2 \right)$ and $a(x)|g(x)|\left( x^{\beta}-1\right) \left(\bmod%
~2\right)$.
\end{theorem}

Considering the theorem above, we can write the type of $\mathcal{C}=\langle
\left(f(x),0\right) ,\left( l(x),g(x)+ua(x)\right) \rangle$ in terms of the
degrees of the polynomials $f(x),~a(x)$ and $g(x)$. Let $\deg
f(x)=t_{1},~\deg g(x)=t_{2}$ and $\deg a(x)=t_{3}$. Then ${\mathcal{C}}$ is
of type
\begin{equation*}
\left( \alpha,\beta;\alpha-t_{4};\beta-t_{2},t_{2}+t_{4}-t_{1}-t_{3}\right)
\end{equation*}
where $d_{1}(x)=\gcd \left( f(x),\dfrac{x^{\beta}-1}{g(x)}l(x)\right) $ and $%
t_{4}=\deg d_{1}(x)$.

\begin{remark}
Let $f(x)=a_{0}+a_{1}x+\ldots +a_{r}x^{r}$ be a polynomial in $R[x]/\left(
x^{n}-1\right).$ And let us define $\widetilde{f(}x)=x^{n-1}f(1/x)$ and the
reciprocal of $f(x)$ to be $f^{\ast }(x)=x^{r}f(1/x\dot{)}.$ Hence, $%
\widetilde{f(}x)=x^{n-1-\deg f}$ $f^{\ast }(x).$ We use often $f$ instead of
$f(x)$ interchangeably throughout the paper. Also note that $%
\deg(f(x))=\deg(f^{\ast }(x))$ if $a_{0}\neq 0$.
\end{remark}

\begin{theorem}
\cite{consta} \label{perp} Let $V=\left( a_{0}a_{1}\ldots a_{\alpha
-1},b_{0}b_{1}\ldots b_{\beta -1}\right) ,~W=\left( d_{0}d_{1}\ldots
d_{\alpha -1},e_{0}e_{1}\ldots e_{\beta -1}\right) \in \mathbb{Z}
_{2}^{\alpha }R^{\beta }$ and $m=lcm(\alpha,\beta)$. Then $V$ is orthogonal
to $W$ and all its cyclic shifts if and only if
\begin{equation*}
G(x)=\left[ ua(x)\widetilde{d(}x)\bmod\left(x^{\alpha }-1\right)\right]
\left( \frac{x^{m}-1}{x^{\alpha }-1}\right) +\left[b(x)\widetilde{e(}x)\bmod%
(x^{\beta}-1) \right]\left( \frac{x^{m}-1}{x^{\beta }-1}\right) =0\bmod%
\left(x^{m}-1\right).
\end{equation*}
\end{theorem}

The following theorem gives the generator polynomials of the dual cyclic
code ${\mathcal{C}}^{\perp}$ of ${\mathbb{Z}_{2}\mathbb{Z}_{2}[u]}$-cyclic
code ${\mathcal{C}}$.

\begin{theorem}
\cite{consta} \label{dual} Let ${\mathcal{C}}=\left\langle \left( f,0\right)
,\left( l,g+ua\right) \right\rangle $ be a cyclic code in $R_{\alpha,\beta}$%
. Then the code
\begin{equation*}
{\mathcal{D}}=\left\langle \left( \left( \frac{x^{\alpha }-1}{d}\right)
,0\right) ,\left( \left( \alpha _{1}\left( Qh_{f}+d_{2}h_{f}\right) \right)
,\left( \theta +u\sigma \right) \right) \right\rangle
\end{equation*}
has type $\left( \alpha ,\beta ;\overline{k_{0}};\overline{k_{1}},\overline{%
k_{2}}\right) $ where
\begin{eqnarray*}
\overline{k_{0}} &=&t_{4} \\
\overline{k_{1}} &=&t_{1}+t_{3}-t_{4} \\
\overline{k_{2}} &=&t_{2}+t_{4}-t_{1}-t_{3}.
\end{eqnarray*}
Furthermore,
\begin{equation*}
{\mathcal{C}}^{\perp }=\left\langle \left(\widetilde{ \left( \frac{x^{\alpha
}-1}{d} \right)},0\right) ,\left( \widetilde{\left(\alpha
_{1}\left(Qh_{f}+d_{2}h_{f}\right) \right)},\widetilde{\left( \theta
+u\sigma \right)}\right) \right\rangle
\end{equation*}

where

\begin{eqnarray*}
d_{1}=\gcd\left(f,\dfrac{x^{\beta }-1}{g}l\right),~d=\gcd \left(
l,f\right),~f=d_{1}d_{2} \\
d_{1}=dQ,~\alpha _{1}l+\alpha_{2}f=d,~l=dd_{4},~\dfrac{x^{\beta }-1}{g}%
l=d_{1}v_{1} \\
\alpha _{3}\dfrac{x^{\beta }-1}{g}l+\alpha _{4}f=d_{1},~\left( \dfrac{%
x^{\beta }-1}{a}l\right) =f\beta _{1} \\
\alpha _{3}\dfrac{x^{\beta }-1}{g}\beta _{1}+\alpha _{4}\dfrac{x^{\beta }-1}{%
a}=\theta,~\alpha _{1}v_{1}+\alpha _{2}\dfrac{x^{\beta }-1}{g}d_{2}=\sigma
\end{eqnarray*}
and $t_{1}=\deg f(x),~t_{2}=\deg g(x),~t_{3}=\deg a(x)$ and $t_{4}=\deg
d_{1}(x)$.
\end{theorem}

\section{The Structure of Self-Dual ${\mathbb{Z}_{2}\mathbb{Z}_{2}[u]}$-Cyclic Codes}

In this section we study the algebraic structure of self-dual ${\mathbb{Z}%
_{2}\mathbb{Z}_{2}[u]}$-cyclic codes using the generator polynomials of the
cyclic code ${\mathcal{C}}$ and its dual ${\mathcal{C}}^{\perp}$. Note that
the code ${\mathcal{C}}$ is self-orthogonal, if ${\mathcal{C}}\subseteq {%
\mathcal{C}}^{\perp}$ and self-dual, if ${\mathcal{C}}={\mathcal{C}}^{\perp}$.

\begin{lemma}
\label{Lemma1} Let $S$ be any commutative ring and  $C=\langle l,k\left( g+ua\right) \rangle $ be any $S$-module. If $k$ is a unit in the ring $S$ with inverse $s$ then ${\mathcal{C}}%
=\langle sl,g+ua\rangle $.
\end{lemma}

\begin{proof}
Let $D=\langle sl,g+ua\rangle.$ Then $k\left( sl,g+ua\right)
=\left(l,k\left( g+ua\right) \right) $ and ${\mathcal{C}}\subseteq D.$ Also,
$s\left( l,k\left( g+ua\right) \right) =\left( sl,g+ua\right) $ and hence $%
D\subseteq {\mathcal{C}}.$ Therefore, ${\mathcal{C}}=D$.
\end{proof}


\begin{theorem}
Let ${\mathcal{C}}=\langle \left( f,0\right) ,\left( l,g+ua\right) \rangle $
be a ${\mathbb{Z}_{2}\mathbb{Z}_{2}[u]}$-cyclic code. Then

\begin{equation*}
{\mathcal{C}}^{\perp }=\left\langle \left( \left( \frac{x^{\alpha }-1}{d}%
\right) ^{\ast },0\right) ,\left( M,x^{i}\left( \theta ^{\ast
}+ux^{j-i}\sigma ^{\ast }\right) \right) \right\rangle
\end{equation*}

where $M=\widetilde{\left( \alpha _{1}\left( Qh_{f}+d_{2}h_{f}\right)
\right) },~i=\beta -1-\deg \theta ,~j=\beta -1-\deg \sigma$.
\end{theorem}

\begin{proof}
From Theorem \ref{dual} we have
\begin{eqnarray*}
{\mathcal{C}}^{\perp } &=&\left\langle \left( \widetilde{\left( \frac{%
x^{\alpha }-1 }{d}\right) },0\right) ,\left( M,\left( \widetilde{\theta
+u\sigma }\right) \right) \right\rangle \\
&=&\left\langle \left( x^{\deg(d)-1}\left( \frac{x^{\alpha }-1}{d}\right)
^{\ast },0\right) ,\left( M,x^{\beta -1-\deg \theta }\theta ^{\ast
}+ux^{\beta -1-\deg \sigma }\sigma ^{\ast }\right) \right\rangle.
\end{eqnarray*}

Also note that
\begin{equation*}
\left( x^{\deg(d)-1}\left( \frac{x^{\alpha }-1}{d}\right) ^{\ast},0\right)
=\left( \left( \frac{x^{\alpha }-1}{d}\right) ^{\ast },0\right).
\end{equation*}
Therefore,
\begin{eqnarray*}
{\mathcal{C}}^{\perp } &=&\left\langle \left( x^{\deg(d)-1}\left( \frac{%
x^{\alpha }-1}{d} \right) ^{\ast },0\right) ,\left( M,x^{\beta -1-\deg
\theta }\theta ^{\ast }+ux^{\beta -1-\deg \sigma }\sigma ^{\ast }\right)
\right\rangle \\
&=&\left\langle \left( \left( \frac{x^{\alpha }-1}{d} \right) ^{\ast
},0\right) ,\left( M,x^{\beta -1-\deg \theta }\theta ^{\ast }+ux^{\beta
-1-\deg \sigma }\sigma ^{\ast }\right) \right\rangle \\
&=&\left\langle \left( \left( \frac{x^{\alpha }-1}{d}\right) ^{\ast
},0\right) ,\left( M,x^{i}\theta ^{\ast }+ux^{j}\sigma ^{\ast }\right)
\right\rangle.
\end{eqnarray*}
Hence

\begin{equation*}
{\mathcal{C}}^{\perp }=\left\langle \left( \left( \frac{x^{\alpha }-1}{d}%
\right) ^{\ast },0\right) ,\left( M,x^{i}\left( \theta ^{\ast
}+ux^{j-i}\sigma ^{\ast }\right) \right) \right\rangle
\end{equation*}

with $M=\widetilde{\left( \alpha _{1}\left( Qh_{f}+d_{2}h_{f}\right)\right) }%
,~i=\beta -1-\deg \theta^{\ast} ,~j=\beta -1-\deg \sigma^{\ast}$. Here, $%
j\geq i$ because $\sigma ^{\ast }|$ $\theta ^{\ast }.$
\end{proof}

Now, let $m=lcm\left(\alpha,\beta \right).$ Since $x^{i}$ is a unit then by
Lemma \ref{Lemma1}, we get that
\begin{equation*}
{\mathcal{C}}^{\perp}=\left\langle\left(\left(\frac{x^{\alpha }-1}{d}%
\right)^{\ast},0\right),\left( Mx^{m-i},\theta ^{\ast }+ux^{j-i}\sigma
^{\ast}\right) \right\rangle.
\end{equation*}%
Consider the cyclic code ${\mathcal{C}}_{1}=\left( \theta
^{\ast}+ux^{j-i}\sigma ^{\ast }\right) $ in $R[x]/\left( x^{\beta
}-1\right). $ Let $h=\theta ^{\ast }+ux^{j-i}\sigma ^{\ast }.$ Then $%
uh=u\theta ^{\ast }\in {\mathcal{C}}_{1}$ and $\left( \dfrac{x^{\beta }-1}{%
\theta ^{\ast }}\right) h=\dfrac{x^{\beta }-1}{\theta ^{\ast }}%
ux^{j-i}\sigma ^{\ast }\in {\mathcal{C}}_{1}.$ Since $\beta $ is odd then $%
\left( x^{\beta }-1\right) $ factors uniquely into irreducible polynomials.
Hence, $\gcd \left( \dfrac{x^{\beta }-1}{\theta ^{\ast }},\theta ^{\ast
}\right) =1.$ Hence,
\begin{eqnarray*}
\dfrac{x^{\beta }-1}{\theta ^{\ast }}f_{1}+\theta ^{\ast }f_{2} &=&1 \\
ux^{j-i}\sigma ^{\ast } &=&\dfrac{x^{\beta }-1}{\theta ^{\ast }}%
ux^{j-i}\sigma ^{\ast }f_{1}+u\sigma ^{\ast }x^{j-i}\theta ^{\ast }f_{2}\in {%
\mathcal{C}}_{1}.
\end{eqnarray*}%
Therefore $u\sigma ^{\ast }\in {\mathcal{C}}_{1}$ and hence $\theta ^{\ast
}\in {\mathcal{C}}_{1}$. So, ${\mathcal{C}}_{1}=\left( \theta ^{\ast
}+ux^{j-i}\sigma ^{\ast }\right) =\left( \theta ^{\ast }+u\sigma ^{\ast
}\right) .$ This implies that $\theta ^{\ast }+ux^{j-i}\sigma ^{\ast
}=q_{1}\left( \theta ^{\ast }+u\sigma ^{\ast }\right) $ and $\theta ^{\ast
}+u\sigma ^{\ast }=q_{2}\left( \theta ^{\ast }+ux^{j-i}\sigma ^{\ast
}\right) .$ So, $\theta ^{\ast }+ux^{j-i}\sigma ^{\ast }=q_{1}q_{2}\left(
\theta ^{\ast }+ux^{j-i}\sigma ^{\ast }\right) $ which implies that $q_{1}$
and $q_{2}$ must be unit and $q_{1}^{-1}=q_{2}.$

In summary we have
\begin{eqnarray*}
{\mathcal{C}}^{\perp } &=&\left\langle \left( \left( \frac{x^{\alpha }-1}{d}%
\right) ^{\ast },0\right) ,\left( Mx^{m-i},\theta ^{\ast }+ux^{j-i}\sigma
^{\ast }\right) \right\rangle \\
&=&\left\langle \left( \left( \frac{x^{\alpha }-1}{d}\right) ^{\ast
},0\right) ,\left( Mx^{m-i},q_{1}\left( \theta ^{\ast }+u\sigma ^{\ast
}\right) \right) \right\rangle.
\end{eqnarray*}
By Lemma \ref{Lemma1}, we give the following theorem.

\begin{theorem}
Let ${\mathcal{C}}=\left\langle \left( f,0\right) ,\left( l,g+ua\right)
\right\rangle $ be a ${\mathbb{Z}_{2}\mathbb{Z}_{2}[u]}$-cyclic code. Then
\begin{eqnarray*}
{\mathcal{C}}^{\perp } &=&\left\langle \left( \left( \frac{x^{\alpha }-1}{d}%
\right) ^{\ast },0\right) ,\left( Mx^{m-i}q_{2},q_{1}\left( \theta ^{\ast
}+u\sigma ^{\ast }\right) \right) \right\rangle \\
&=&\left\langle \left( \left( \frac{x^{\alpha }-1}{d}\right) ^{\ast
},0\right) ,\left( Q,\theta ^{\ast }+u\sigma ^{\ast }\right) \right\rangle
\end{eqnarray*}
where $Q=$ $Mx^{m-i}q_{2}~mod\left( \dfrac{x^{\alpha }-1}{d}\right) ^{\ast
},~M=\widetilde{\left( \alpha _{1}\left( Qh_{f}+d_{2}h_{f}\right) \right) }%
,~i=\beta -1-\deg \theta^{\ast} ,~j=\beta -1-\deg \sigma^{\ast} $ and $%
\theta ^{\ast }+u\sigma ^{\ast }=q_{2}\left( \theta ^{\ast }+ux^{j-i}\sigma
^{\ast }\right).$
\end{theorem}

\begin{theorem}
Let ${\mathcal{C}}=\left\langle \left( f,0\right) ,\left( l,g+ua\right)
\right\rangle$ be a ${\mathbb{Z}_{2}\mathbb{Z}_{2}[u]}$-cyclic code. Then ${%
\mathcal{C}}$ is self-dual iff $f=\left( \frac{x^{\alpha }-1}{d}\right)
^{\ast },~l=Q,~g=\theta ^{\ast } $ and $a=\sigma ^{\ast }.$
\end{theorem}

\begin{corollary}
Let ${\mathcal{C}}$ be a self-dual ${\mathbb{Z}_{2}\mathbb{Z}_{2}[u]}$%
-cyclic code in $R_{\alpha,\beta}$ generated by $\left(\left(f,0\right),%
\left(l,g+ua\right)\right)$. Then

\begin{equation*}
\beta+\frac{\alpha}{2}=t_{1}+t_{2}+t_{3}
\end{equation*}

where $t_{1}=\deg f(x),~t_{2}=\deg g(x),~t_{3}=\deg a(x)$.
\end{corollary}

\begin{proof}
Since ${\mathcal{C}}$ is self-dual, then the code ${\mathcal{C}}$ and its
dual have the same type. Hence, we have
\begin{eqnarray*}
&&k_{0}=\alpha -t_{4}=\overline{k_{0}}=t_{4}\Rightarrow \alpha =2t_{4} \\
&&k_{1}=\beta -t_{2}=\overline{k_{1}}=t_{1}+t_{3}-t_{4}\Rightarrow \beta
=t_{1}+t_{2}+t_{3}-t_{4} \\
&&\beta +\frac{\alpha }{2}=t_{1}+t_{2}+t_{3}.
\end{eqnarray*}
\end{proof}

We know from \cite{consta} that if $\mathcal{C}=\langle \left( f(x),0\right)
,\left( l(x),g(x)+ua(x)\right) \rangle $ is a cyclic code in $%
R_{\alpha,\beta}$ with $f(x)|\left(x^{\alpha }-1\right)(\bmod~2),$ and $%
a(x)|g(x)|\left( x^{\beta }-1\right) \left(\bmod~2\right),$ and $l(x)$ is a
binary polynomial satisfying $\deg (l(x))<\deg (f(x)),$ $f(x)|\left( \dfrac{%
x^{\beta }-1}{a(x)} \right) l(x)\left(\bmod~u \right)$ and $f(x)\neq\left(
\dfrac{x^{\beta }-1}{ a(x)}\right)l(x)\left(\bmod~u\right)$ and $%
f(x)h_{f}(x)=x^{\alpha}-1$, $g(x)h_{g}(x)=x^{\beta}-1$, $g(x)=a(x)b(x)$,
then
\begin{equation*}
S_1=\bigcup_{i=0}^{\alpha-deg(f)-1}\left\{x^i \ast \left(
f(x),0\right)\right\},
\end{equation*}

\begin{equation*}
S_{2}=\bigcup_{i=0}^{\beta-deg(g)-1}\left\{ x^{i}\ast
\left(l(x),g(x)+ua(x)\right) \right\}
\end{equation*}

and
\begin{equation*}
S_{3}=\bigcup_{i=0}^{deg(g)-deg(a)-1}\left\{ x^{i}\ast
\left(h_g(x)l(x),uh_{g}(x)a(x)\right) \right\} .
\end{equation*}
\begin{equation*}
S=S_{1}\cup S_{2}\cup S_{3}
\end{equation*}
forms a minimal spanning set for $\mathcal{C}$.

Similarly, we give the following theorem that determines the minimal
spanning sets for the dual cyclic code ${\mathcal{C}}^{\perp}$.

\begin{theorem}
Let ${\mathcal{C}}^{\perp }=\left\langle \left( \left( \frac{x^{\alpha }-1}{d%
}\right)^{\ast },0\right) ,\left( Q,\theta ^{\ast }+u\sigma ^{\ast }\right)
\right\rangle$. Then
\begin{equation*}
T_{1}=\bigcup_{i=0}^{\alpha -\deg(h_{d}^{\ast })-1}\left\{ x^{i}\ast \left(
h_{d}^{\ast },0\right) \right\} ,
\end{equation*}

\begin{equation*}
T_{2}=\bigcup_{i=0}^{\beta -\deg (\theta ^{\ast })-1}\left\{ x^{i}\ast
\left( Q,\theta ^{\ast }+u\sigma ^{\ast }\right) \right\}
\end{equation*}%
and
\begin{equation*}
T_{3}=\bigcup_{i=0}^{\deg (\theta ^{\ast })-\deg (\sigma ^{\ast })-1}\left\{
x^{i}\ast \left( J^{\ast }Q,uJ^{\ast }\sigma ^{\ast }\right) \right\}
\end{equation*}%
forms a minimal spanning set for ${\mathcal{C}}^{\perp }$, where $%
h_{d}^{\ast }=\left( \frac{x^{\alpha }-1}{d}\right) ^{\ast }$, and $J\theta
=x^{\beta }-1$. Furthermore, $\mathcal{{\mathcal{C}}^{\perp }}$ has $%
2^{\alpha -\deg (h_{d}^{\ast })}4^{\beta -\deg (\theta ^{\ast })}2^{\deg
(\theta ^{\ast })-\deg (\sigma ^{\ast })}$ codewords.
\end{theorem}

\begin{proof}
See Theorem 14 in \cite{consta}
\end{proof}

Hence, we can conclude with the following results.

\begin{corollary}
If ${\mathcal{C}}$ is a self-dual ${\mathbb{Z}_{2}\mathbb{Z}_{2}[u]}$-cyclic
code then,
\begin{eqnarray*}
\alpha -deg(h_{d}) &=&\alpha -deg(f)\Rightarrow \alpha -(\alpha
-deg(d))=\alpha -deg(f) \\
&\Rightarrow &\alpha =deg(f)+deg(d) \\
deg(g) &=&deg(\theta )\Rightarrow deg(g)=\beta +deg(d_{1})-deg(f)-deg(a) \\
deg(a) &=&deg(\sigma )\Rightarrow deg(a)=\beta +deg(d)-deg(d_{1})-deg(g).
\end{eqnarray*}
\end{corollary}

\begin{corollary}
If ${\mathcal{C}}$ is a self-dual ${\mathbb{Z}_{2}\mathbb{Z}_{2}[u]}$-cyclic
code then it is orthogonal to itself. Hence we have
\begin{eqnarray*}
&&uf\widetilde{f}\frac{x^{m}-1}{x^{\alpha }-1} =0\left( \bmod~x^{m}-1\right)
\Rightarrow (x^{\alpha }-1)|f\widetilde{f}\Rightarrow
h_{f}|x^{\alpha-\deg(f)-1}f^{\ast } \\
&&uf\widetilde{l}\frac{x^{m}-1}{x^{\alpha }-1} =0\left( \bmod~x^{m}-1\right)
\Rightarrow (x^{\alpha }-1)|f\widetilde{l}\Rightarrow
h_{f}|x^{\alpha-\deg(l)-1}l^{\ast } \\
&&ul\widetilde{l }\frac{x^{m}-1}{x^{\alpha }-1}+(g+ua)\widetilde{(g+ua)}%
\frac{x^{m}-1}{ x^{\beta }-1} =0\left( \bmod~x^{m}-1\right) \\
&&ul\widetilde{l}\frac{x^{m}-1}{x^{\alpha }-1}+[g\widetilde{g }+u(a%
\widetilde{g }+g\widetilde{a })]\frac{x^{m}-1}{x^{\beta }-1} =0\left( \bmod%
~x^{m}-1\right) \\
&&ul\widetilde{l}\frac{x^{m}-1}{x^{\alpha }-1}+u(a\widetilde{g}+g\widetilde{%
a })\frac{x^{m}-1}{x^{\beta }-1}+g\widetilde{g }\frac{x^{m}-1}{x^{\beta }-1}
=0\left( \bmod~x^{m}-1\right)\Rightarrow \\
&&(x^{\beta }-1)|g\widetilde{g}\Rightarrow h_{g}|x^{\beta-\deg(g)-1}g^{\ast
} ~\text{and} \\
&&(x^{\alpha }-1)|l\widetilde{l} \Rightarrow (x^{\alpha
}-1)|x^{\alpha-\deg(l)-1}ll^{\ast }~\text{and}~(x^{\beta }-1)|(a\widetilde{g}%
+g\widetilde{a})~\text{or}~ \\
&&l\widetilde{l}\frac{x^{m}-1}{x^{\alpha }-1} =(a\widetilde{g}+g\widetilde{a}%
) \frac{x^{m}-1}{x^{\beta }-1}\left( \bmod~x^{m}-1\right) \Rightarrow \\
&&l\widetilde{l}\left(x^{\beta }-1\right) =(a\widetilde{g}+g\widetilde{a}%
)\left(x^{\alpha }-1\right)\left( \bmod~x^{m}-1\right).
\end{eqnarray*}
\end{corollary}

\subsection{Examples of Self-Dual ${\mathbb{Z}_{2}\mathbb{Z}_{2}[u]}$-Cyclic
Codes}

In this section of the paper, we give some examples of self-dual ${\mathbb{Z}%
_{2}\mathbb{Z}_{2}[u]}$-cyclic codes. Furthermore, we list some of them with
their binary parameters in Table \ref{table}.

\begin{example}
Let $\mathcal{C}$ be a $\mathbb{Z}_{2}\mathbb{Z}_{2}[u]$-cyclic code in $%
\mathbb{Z}_{2}[x]/\langle x^{14}-1\rangle\times R[x]/\langle x^{21}-1\rangle$
and let $\mathcal{C}$ is generated by $\left((f(x),0),(l(x),g(x)+ua(x))%
\right) $ with
\begin{eqnarray*}
f(x) &=&1+x+x^3+x^7+x^8+x^{10}, \\
g(x) &=&1+x+x^2+x^3+x^7+x^9+x^{11}+x^{12}, \\
a(x) &=&1 + x^2 + x^4 + x^5 + x^6, \\
l(x) &=&1+x^3+x^5+x^6.
\end{eqnarray*}
Therefore, we can calculate the following polynomials.
\begin{eqnarray*}
f(x)h_{f}(x) &=&x^{14}-1\Rightarrow h_{f}(x)=1 + x + x^2 + x^4, \\
g(x)h_{g}(x) &=&x^{21}-1\Rightarrow h_{g}(x)=1 + x + x^4 + x^5 + x^7 + x^8 +
x^9, \\
a(x)h_{a}(x) &=&x^{21}-1\Rightarrow
h_{a}(x)=1+x^2+x^5+x^8+x^9+x^{12}+x^{14}+x^{15}, \\
d(x)&=&gcd(f(x),l(x))=1 + x^2 + x^3 + x^4,
d_{1}(x)=gcd(f(x),h_{g}(x)l(x))=x^7-1, \\
f(x)&=&d_{1}(x)d_{2}(x)\Rightarrow d_{2}(x)=1 + x + x^3,
d_{1}(x)=d(x)Q(x)\Rightarrow Q(x)=1+ x^2 + x^3, \\
d(x)&=&\alpha_{1}(x)l(x)+\alpha_{2}(x)f(x)\Rightarrow\alpha_{1}(x)=1 + x +
x^4 + x^5, \alpha_{2}(x)=x, \\
d_{1}(x)&=&\alpha_{3}(x)h_{g}(x)l(x)+\alpha_{4}(x)f(x)\Rightarrow%
\alpha_{3}(x)=1+x^2, \alpha_{4}(x)=x + x^4 + x^7, \\
h_{a}(x)l(x)&=&\beta_{1}(x)f(x)\Rightarrow
\beta_{1}(x)=1+x+x^4+x^6+x^8+x^{11}, \\
h_{g}(x)l(x)&=&\vartheta_{1}(x)d_{1}(x)\Rightarrow
\vartheta_{1}(x)=1+x+x^3+x^8, \\
\theta(x)&=&\alpha_{3}(x)\beta_{1}(x)h_{g}(x)+\alpha_{4}(x)h_{a}(x)%
\Rightarrow \theta(x)=1+x+x^3+x^5+x^9+x^{10}+x^{11}+x^{12}, \\
\sigma(x)&=&\alpha_{1}(x)\vartheta_{1}(x)+\alpha_{2}(x)d_{2}(x)h_{g}(x)
\Rightarrow \sigma(x)=1 + x + x^2 + x^4 + x^6.
\end{eqnarray*}
Hence, using the generator sets above, we can write the generator and
parity-check matrices for $\mathcal{C}$ as follows.
\begin{eqnarray*}
G=\left(
\begin{smallmatrix}
1 & 1 & 0 & 1 & 0 & 0 & 0 & 1 & 1 & 0 & 1 & 0 & 0 & 0 & 0 & 0 & 0 & 0 & 0 & 0
& 0 & 0 & 0 & 0 & 0 & 0 & 0 & 0 & 0 & 0 & 0 & 0 & 0 & 0 & 0 \\
0 & 1 & 1 & 0 & 1 & 0 & 0 & 0 & 1 & 1 & 0 & 1 & 0 & 0 & 0 & 0 & 0 & 0 & 0 & 0
& 0 & 0 & 0 & 0 & 0 & 0 & 0 & 0 & 0 & 0 & 0 & 0 & 0 & 0 & 0 \\
0 & 0 & 1 & 1 & 0 & 1 & 0 & 0 & 0 & 1 & 1 & 0 & 1 & 0 & 0 & 0 & 0 & 0 & 0 & 0
& 0 & 0 & 0 & 0 & 0 & 0 & 0 & 0 & 0 & 0 & 0 & 0 & 0 & 0 & 0 \\
0 & 0 & 0 & 1 & 1 & 0 & 1 & 0 & 0 & 0 & 1 & 1 & 0 & 1 & 0 & 0 & 0 & 0 & 0 & 0
& 0 & 0 & 0 & 0 & 0 & 0 & 0 & 0 & 0 & 0 & 0 & 0 & 0 & 0 & 0 \\
1 & 0 & 0 & 1 & 0 & 1 & 1 & 0 & 0 & 0 & 0 & 0 & 0 & 0 & 1+u & 1 & 1+u & 1 & u
& u & u & 1 & 0 & 1 & 0 & 1 & 1 & 0 & 0 & 0 & 0 & 0 & 0 & 0 & 0 \\
0 & 1 & 0 & 0 & 1 & 0 & 1 & 1 & 0 & 0 & 0 & 0 & 0 & 0 & 0 & 1+u & 1 & 1+u & 1
& u & u & u & 1 & 0 & 1 & 0 & 1 & 1 & 0 & 0 & 0 & 0 & 0 & 0 & 0 \\
0 & 0 & 1 & 0 & 0 & 1 & 0 & 1 & 1 & 0 & 0 & 0 & 0 & 0 & 0 & 0 & 1+u & 1 & 1+u
& 1 & u & u & u & 1 & 0 & 1 & 0 & 1 & 1 & 0 & 0 & 0 & 0 & 0 & 0 \\
0 & 0 & 0 & 1 & 0 & 0 & 1 & 0 & 1 & 1 & 0 & 0 & 0 & 0 & 0 & 0 & 0 & 1+u & 1
& 1+u & 1 & u & u & u & 1 & 0 & 1 & 0 & 1 & 1 & 0 & 0 & 0 & 0 & 0 \\
0 & 0 & 0 & 0 & 1 & 0 & 0 & 1 & 0 & 1 & 1 & 0 & 0 & 0 & 0 & 0 & 0 & 0 & 1+u
& 1 & 1+u & 1 & u & u & u & 1 & 0 & 1 & 0 & 1 & 1 & 0 & 0 & 0 & 0 \\
0 & 0 & 0 & 0 & 0 & 1 & 0 & 0 & 1 & 0 & 1 & 1 & 0 & 0 & 0 & 0 & 0 & 0 & 0 &
1+u & 1 & 1+u & 1 & u & u & u & 1 & 0 & 1 & 0 & 1 & 1 & 0 & 0 & 0 \\
0 & 0 & 0 & 0 & 0 & 0 & 1 & 0 & 0 & 1 & 0 & 1 & 1 & 0 & 0 & 0 & 0 & 0 & 0 & 0
& 1+u & 1 & 1+u & 1 & u & u & u & 1 & 0 & 1 & 0 & 1 & 1 & 0 & 0 \\
0 & 0 & 0 & 0 & 0 & 0 & 0 & 1 & 0 & 0 & 1 & 0 & 1 & 1 & 0 & 0 & 0 & 0 & 0 & 0
& 0 & 1+u & 1 & 1+u & 1 & u & u & u & 1 & 0 & 1 & 0 & 1 & 1 & 0 \\
1 & 0 & 0 & 0 & 0 & 0 & 0 & 0 & 1 & 0 & 0 & 1 & 0 & 1 & 0 & 0 & 0 & 0 & 0 & 0
& 0 & 0 & 1+u & 1 & 1+u & 1 & u & u & u & 1 & 0 & 1 & 0 & 1 & 1 \\
1 & 0 & 0 & 1 & 0 & 0 & 0 & 1 & 0 & 0 & 1 & 0 & 0 & 0 & u & u & u & u & 0 & u
& u & u & 0 & 0 & u & u & 0 & u & 0 & u & 0 & 0 & 0 & 0 & 0 \\
0 & 1 & 0 & 0 & 1 & 0 & 0 & 0 & 1 & 0 & 0 & 1 & 0 & 0 & 0 & u & u & u & u & 0
& u & u & u & 0 & 0 & u & u & 0 & u & 0 & u & 0 & 0 & 0 & 0 \\
0 & 0 & 1 & 0 & 0 & 1 & 0 & 0 & 0 & 1 & 0 & 0 & 1 & 0 & 0 & 0 & u & u & u & u
& 0 & u & u & u & 0 & 0 & u & u & 0 & u & 0 & u & 0 & 0 & 0 \\
0 & 0 & 0 & 1 & 0 & 0 & 1 & 0 & 0 & 0 & 1 & 0 & 0 & 1 & 0 & 0 & 0 & u & u & u
& u & 0 & u & u & u & 0 & 0 & u & u & 0 & u & 0 & u & 0 & 0 \\
1 & 0 & 0 & 0 & 1 & 0 & 0 & 1 & 0 & 0 & 0 & 1 & 0 & 0 & 0 & 0 & 0 & 0 & u & u
& u & u & 0 & u & u & u & 0 & 0 & u & u & 0 & u & 0 & u & 0 \\
0 & 1 & 0 & 0 & 0 & 1 & 0 & 0 & 1 & 0 & 0 & 0 & 1 & 0 & 0 & 0 & 0 & 0 & 0 & u
& u & u & u & 0 & u & u & u & 0 & 0 & u & u & 0 & u & 0 & u%
\end{smallmatrix}
\right)
\end{eqnarray*}

\begin{equation*}
H=\left(
\begin{smallmatrix}
0 & 0 & 0 & 1 & 1 & 0 & 1 & 0 & 0 & 0 & 1 & 1 & 0 & 1 & 0 & 0 & 0 & 0 & 0 & 0
& 0 & 0 & 0 & 0 & 0 & 0 & 0 & 0 & 0 & 0 & 0 & 0 & 0 & 0 & 0 \\
1 & 0 & 0 & 0 & 1 & 1 & 0 & 1 & 0 & 0 & 0 & 1 & 1 & 0 & 0 & 0 & 0 & 0 & 0 & 0
& 0 & 0 & 0 & 0 & 0 & 0 & 0 & 0 & 0 & 0 & 0 & 0 & 0 & 0 & 0 \\
0 & 1 & 0 & 0 & 0 & 1 & 1 & 0 & 1 & 0 & 0 & 0 & 1 & 1 & 0 & 0 & 0 & 0 & 0 & 0
& 0 & 0 & 0 & 0 & 0 & 0 & 0 & 0 & 0 & 0 & 0 & 0 & 0 & 0 & 0 \\
1 & 0 & 1 & 0 & 0 & 0 & 1 & 1 & 0 & 1 & 0 & 0 & 0 & 1 & 0 & 0 & 0 & 0 & 0 & 0
& 0 & 0 & 0 & 0 & 0 & 0 & 0 & 0 & 0 & 0 & 0 & 0 & 0 & 0 & 0 \\
0 & 0 & 1 & 0 & 0 & 1 & 1 & 1 & 1 & 1 & 0 & 1 & 1 & 0 & 0 & 0 & 0 & 0 & 0 & 0
& 0 & 0 & 1 & 1 & 1 & 1 & 0 & 0 & u & 1 & u & 1 & u & 1+u & 1+u \\
0 & 0 & 0 & 1 & 0 & 0 & 1 & 1 & 1 & 1 & 1 & 0 & 1 & 1 & 1+u & 0 & 0 & 0 & 0
& 0 & 0 & 0 & 0 & 1 & 1 & 1 & 1 & 0 & 0 & u & 1 & u & 1 & u & 1+u \\
1 & 0 & 0 & 0 & 1 & 0 & 0 & 1 & 1 & 1 & 1 & 1 & 0 & 1 & 1+u & 1+u & 0 & 0 & 0
& 0 & 0 & 0 & 0 & 0 & 1 & 1 & 1 & 1 & 0 & 0 & u & 1 & u & 1 & u \\
1 & 1 & 0 & 0 & 0 & 1 & 0 & 0 & 1 & 1 & 1 & 1 & 1 & 0 & u & 1+u & 1+u & 0 & 0
& 0 & 0 & 0 & 0 & 0 & 0 & 1 & 1 & 1 & 1 & 0 & 0 & u & 1 & u & 1 \\
0 & 1 & 1 & 0 & 0 & 0 & 1 & 0 & 0 & 1 & 1 & 1 & 1 & 1 & 1 & u & 1+u & 1+u & 0
& 0 & 0 & 0 & 0 & 0 & 0 & 0 & 1 & 1 & 1 & 1 & 0 & 0 & u & 1 & u \\
1 & 0 & 1 & 1 & 0 & 0 & 0 & 1 & 0 & 0 & 1 & 1 & 1 & 1 & u & 1 & u & 1+u & 1+u
& 0 & 0 & 0 & 0 & 0 & 0 & 0 & 0 & 1 & 1 & 1 & 1 & 0 & 0 & u & 1 \\
1 & 1 & 0 & 1 & 1 & 0 & 0 & 0 & 1 & 0 & 0 & 1 & 1 & 1 & 1 & u & 1 & u & 1+u
& 1+u & 0 & 0 & 0 & 0 & 0 & 0 & 0 & 0 & 1 & 1 & 1 & 1 & 0 & 0 & u \\
1 & 1 & 1 & 0 & 1 & 1 & 0 & 0 & 0 & 1 & 0 & 0 & 1 & 1 & u & 1 & u & 1 & u &
1+u & 1+u & 0 & 0 & 0 & 0 & 0 & 0 & 0 & 0 & 1 & 1 & 1 & 1 & 0 & 0 \\
1 & 1 & 1 & 1 & 0 & 1 & 1 & 0 & 0 & 0 & 1 & 0 & 0 & 1 & 0 & u & 1 & u & 1 & u
& 1+u & 1+u & 0 & 0 & 0 & 0 & 0 & 0 & 0 & 0 & 1 & 1 & 1 & 1 & 0 \\
0 & 0 & 1 & 1 & 0 & 0 & 0 & 0 & 0 & 1 & 1 & 0 & 0 & 0 & 0 & 0 & 0 & 0 & u & u
& u & u & 0 & u & u & u & 0 & 0 & u & u & 0 & u & 0 & u & 0 \\
0 & 0 & 0 & 1 & 1 & 0 & 0 & 0 & 0 & 0 & 1 & 1 & 0 & 0 & 0 & 0 & 0 & 0 & 0 & u
& u & u & u & 0 & u & u & u & 0 & 0 & u & u & 0 & u & 0 & u \\
0 & 0 & 0 & 0 & 1 & 1 & 0 & 0 & 0 & 0 & 0 & 1 & 1 & 0 & u & 0 & 0 & 0 & 0 & 0
& u & u & u & u & 0 & u & u & u & 0 & 0 & u & u & 0 & u & 0 \\
0 & 0 & 0 & 0 & 0 & 1 & 1 & 0 & 0 & 0 & 0 & 0 & 1 & 1 & 0 & u & 0 & 0 & 0 & 0
& 0 & u & u & u & u & 0 & u & u & u & 0 & 0 & u & u & 0 & u \\
1 & 0 & 0 & 0 & 0 & 0 & 1 & 1 & 0 & 0 & 0 & 0 & 0 & 1 & u & 0 & u & 0 & 0 & 0
& 0 & 0 & u & u & u & u & 0 & u & u & u & 0 & 0 & u & u & 0 \\
1 & 1 & 0 & 0 & 0 & 0 & 0 & 1 & 1 & 0 & 0 & 0 & 0 & 0 & 0 & u & 0 & u & 0 & 0
& 0 & 0 & 0 & u & u & u & u & 0 & u & u & u & 0 & 0 & u & u%
\end{smallmatrix}
\right)
\end{equation*}
So, it can be easily checked that $H$ is orthogonal to $G$ and their
cardinalities are the same and equal to $2^{4}\cdot 2^{2.9}\cdot 2^{6}=2^{28}
$. It can be also checked that both $G$ and $H$ are orthogonal to
themselves. Then $\mathcal{C}$ is a self-dual ${\mathbb{Z}_{2}\mathbb{Z}%
_{2}[u]}$-cyclic code of type $(14,21;7;9,3)$.
\end{example}

\begin{table}[h!]
\caption{Table of self-dual ${\mathbb{Z}_{2}\mathbb{Z}_{2}[u]}$-cyclic
codes.}
\label{table}\centering
\par
{\rowcolors{3}{red!15!}{cyan!15!}
\begin{tabularx}{\textwidth}{||X||c||c||}
\hline\hline
Generators & $\mathbb{Z}_{2}\mathbb{Z}_{2}[u]-$type & Binary Image  \\ \hline\hline
$f(x)=1+x+x^3+x^7+x^8+x^{10},~l(x)=1+x^3+x^5+x^6,~g(x)=1+x^3+x^6+x^{12},~a(x)=1+x+x^2+x^4+x^6$ & $
[14,21;7;9,3]$ & $[56,28,6]$ \\ \hline\hline
$f(x)=1+x^2+x^3+x^7+x^9+x^{10},~l(x)=1+x+x^3+x^6,~g(x)=1+x+x^4+x^8+x^9+x^{11}+x^{12}+x^{14}+x^{17}+x^{18}+x^{19}+x^{20},~a(x)=1+x+x^2+x^4+x^7+x^8+x^9+x^{10}+x^{12}$ & $[14,35;7;15,15]$ & $[84,52,6]$  \\ \hline\hline
$f(x)=1+x+x^3+x^{14}+x^{15}+x^{17},~l(x)=1+x^3+x^5+x^6+x^7+x^{10}+x^{12}+x^{13}~,g(x)=1+x^{10}+x^{15}+x^{20},~a(x)=1+x^2+x^3+x^4+x^5+x^8+x^{10}+x^{11}+x^{12}$
& $[28,35;7;15,5]$ & $[98,42,6]$  \\ \hline\hline
\end{tabularx}
}
\end{table}

\section{Quantum Codes Over $\Z_{2}^{\alpha}\times R^{\beta}$}

In this section, we use cyclic codes over $\Z_{2}^{\alpha}\times R^{\beta}$ to obtain self-orthogonal codes over ${\mathbb{Z}}_{2}$. Using these self-orthogonal
codes, we determine the parameters of the corresponding quantum codes.

Recall that, for $r_{1}+uq_{1}=a\in R$, $%
r_{1},q_{1}\in \mathbb{Z}_{2}$, the Gray map is defined  as follows.
\begin{align*}
& \qquad \qquad \qquad \qquad \Phi : \mathbb{Z}_{2}^{\alpha }\times R^{\beta
}\rightarrow \mathbb{Z}_{2}^{n}\text{ } \\
& \Phi \left( x_{0},\ldots x_{\alpha -1},r_{0}+uq_{0},\ldots r_{\beta
-1}+uq_{\beta -1}\right) \\
& \qquad \qquad =\left( x_{0},\ldots x_{r-1},q_{0},\ldots ,q_{\beta
-1},r_{0}+q_{0},\ldots ,r_{\beta -1}+q_{\beta -1}\right)
\end{align*}
where $n=\alpha +2\beta$. The map $\Phi $ is an isometry which transforms
the Lee distance in $\mathbb{\ Z}_{2}^{\alpha }\times R^{\beta }$ to the
Hamming distance in $\mathbb{Z}_{2}^{n}$. The Hamming weight of any codeword
is defined as the number of its nonzero entries. The Hamming distance
between two codewords is the Hamming weight of their difference.
Additionally, the Lee distance for the codes on a $q$-ary alphabet is the
Lee weight of their differences where the Lee weight of an element $a\in
\mathbb{Z}_q$ is defined as the minimum absolute distance from the zero
element to $a$ or $q{-}a$. Furthermore, we can define the Lee weights of
elements in $R$ as, $wt_{L}(0)=0,~wt_{L}(1)=1,~wt_{L}(u)=2$ and $%
wt_{L}(1+u)=1$.

Moreover, for any $\mathbb{Z}_{2}\mathbb{Z}_{2}[u]$-linear code $\mathcal{C}%
, $ we have that $\Phi \left(\mathcal{C}\right) $ is a binary linear code as
well. This property is not valid for the $\mathbb{Z}_{2}\mathbb{Z}_{4}-$%
additive codes. For any codeword $v=(v_{1},v_2)\in {\mathcal{C}}$, we always
have
\begin{equation*}
wt(v)=wt_{H}(v_{1})+wt_{L}(v_{2})
\end{equation*}
where $wt_{H}(v_{1})$ is the Hamming of weight of $v_{1}$ and $wt_{L}(v_{2})$
is the Lee weight of $v_{2}.$

The binary image $C=\Phi (\mathcal{C})$ of a $\mathbb{Z}_{2}\mathbb{Z}%
_{2}[u] $-linear code $\mathcal{C}$ of type $\left( \alpha ,\beta
;k_{0},k_{1},k_{2}\right) $ is a binary linear code of length $n=\alpha
+2\beta $ and size $2^{2k_{1}+k_{0}+k_{2}}$. It is also called a $\mathbb{Z}%
_{2}\mathbb{Z}_{2}[u]$-linear code.

Let $v$ and $w$ are two codewords of a $\mathbb{Z}_{2}\mathbb{Z}_{2}[u]$%
-linear code. These codewords may be orthogonal to each other but the binary
parts may not be orthogonal. For example, $\left( 1~1|1+u~u\right) $ and $%
\left( 0~1|u~u\right) $ are orthogonal in $\mathbb{Z}_{2}^2 \times R^2$
whereas the binary or $R$-components are not orthogonal. However in \cite%
{ism}, it was proved the binary image of a self-dual ${\mathbb{Z}_{2}%
\mathbb{Z}_{2}[u]}$-linear code under $\Phi$ is again a self-dual. So, $\Phi
$ preserves orthogonality. Therefore, we can use this result to obtain a
quantum code via CSS method. The following theorem defines the
Calderbank-Shor-Steane (CSS)code construction.

\begin{theorem}\label{css}
\cite{shor,steane} Let ${\mathcal{C}}_{1}$ and ${\mathcal{C}}_{2}$ be two
linear codes with parameters $[n, k_1, d_1]_q$ and $[n,k_2,d_2]_q$ such that
${\mathcal{C}}_{2}\subseteq {\mathcal{C}}_{1}$. Then there exists a $[[n,
k_1-k_2, d]]_q$ stabilizer code with minimum distance $d$ that is pure to $%
min\{d_1,d_{2}^{\perp}\}$ where $d_{2}^{\perp}$ is the minimum distance of
the dual code ${\mathcal{C}}_{2}^{\perp}$.
\end{theorem}

A special case that interests us is when ${\mathcal{C}}_{1} = {\mathcal{C}}%
_{2}$ and also the code is self-orthogonal. This is particularly interesting
because one can easily find codes which satisfy the self-orthogonality
condition.

\begin{theorem}
If ${\mathcal{C}}$ is a self-orthogonal cyclic code in $R_{\alpha,\beta}$ of
type $(\alpha,\beta;k_{0},k_{1},k_{2})$ with the minimum distance $d$, then
there exists an $[[N, N-2K, d^{\perp}]]$ stabilizer code in ${\mathbb{Z}}%
_{2}^{\alpha+2\beta}$ where $N=\alpha+2\beta,~K=k_{0}+k_{2}+2k_{1}$ and $%
d^{\perp}$ is the minimum distance of the dual code ${\mathcal{C}}^{\perp}$.
\end{theorem}

\begin{proof}
Taking $\C_{2}=\C$ and $\C_{1}=\C^{\perp}$ in Theorem \ref{css} we have the result.

\end{proof}
\begin{example}
Let $\mathcal{C}=\left\langle(f(x),0),(l(x),g(x)+ua(x))\right\rangle $ be a
cyclic code in $\mathbb{Z}_{2}[x]/\langle x^{14}-1\rangle\times R[x]/\langle
x^{21}-1\rangle$ where
\begin{eqnarray*}
f(x) &=&1+x+x^3+x^7+x^8+x^{10}, \\
g(x) &=&1+x+x^3+x^6+x^7+x^{10}+x^{13}+x^{15}, \\
a(x) &=&1 + x^2 + x^4 + x^5 + x^6, \\
l(x) &=&1 + x^4 + x^6 + x^8.
\end{eqnarray*}

Therefore, we can write the generator and parity-check matrices for ${%
\mathcal{C}}$ using the same way above. Then we have,
\begin{eqnarray*}
G=\left(
\begin{smallmatrix}
1 & 1 & 0 & 1 & 0 & 0 & 0 & 1 & 1 & 0 & 1 & 0 & 0 & 0 & 0 & 0 & 0 & 0 & 0 & 0
& 0 & 0 & 0 & 0 & 0 & 0 & 0 & 0 & 0 & 0 & 0 & 0 & 0 & 0 & 0 \\
0 & 1 & 1 & 0 & 1 & 0 & 0 & 0 & 1 & 1 & 0 & 1 & 0 & 0 & 0 & 0 & 0 & 0 & 0 & 0
& 0 & 0 & 0 & 0 & 0 & 0 & 0 & 0 & 0 & 0 & 0 & 0 & 0 & 0 & 0 \\
0 & 0 & 1 & 1 & 0 & 1 & 0 & 0 & 0 & 1 & 1 & 0 & 1 & 0 & 0 & 0 & 0 & 0 & 0 & 0
& 0 & 0 & 0 & 0 & 0 & 0 & 0 & 0 & 0 & 0 & 0 & 0 & 0 & 0 & 0 \\
0 & 0 & 0 & 1 & 1 & 0 & 1 & 0 & 0 & 0 & 1 & 1 & 0 & 1 & 0 & 0 & 0 & 0 & 0 & 0
& 0 & 0 & 0 & 0 & 0 & 0 & 0 & 0 & 0 & 0 & 0 & 0 & 0 & 0 & 0 \\
1 & 0 & 0 & 0 & 1 & 0 & 1 & 0 & 1 & 0 & 0 & 0 & 0 & 0 & 1+u & 1 & u & 1 & u
& u & 1+u & 1 & 0 & 0 & 1 & 0 & 0 & 1 & 0 & 1 & 0 & 0 & 0 & 0 & 0 \\
0 & 1 & 0 & 0 & 0 & 1 & 0 & 1 & 0 & 1 & 0 & 0 & 0 & 0 & 0 & 1+u & 1 & u & 1
& u & u & 1+u & 1 & 0 & 0 & 1 & 0 & 0 & 1 & 0 & 1 & 0 & 0 & 0 & 0 \\
0 & 0 & 1 & 0 & 0 & 0 & 1 & 0 & 1 & 0 & 1 & 0 & 0 & 0 & 0 & 0 & 1+u & 1 & u
& 1 & u & u & 1+u & 1 & 0 & 0 & 1 & 0 & 0 & 1 & 0 & 1 & 0 & 0 & 0 \\
0 & 0 & 0 & 1 & 0 & 0 & 0 & 1 & 0 & 1 & 0 & 1 & 0 & 0 & 0 & 0 & 0 & 1+u & 1
& u & 1 & u & u & 1+u & 1 & 0 & 0 & 1 & 0 & 0 & 1 & 0 & 1 & 0 & 0 \\
0 & 0 & 0 & 0 & 1 & 0 & 0 & 0 & 1 & 0 & 1 & 0 & 1 & 0 & 0 & 0 & 0 & 0 & 1+u
& 1 & u & 1 & u & u & 1+u & 1 & 0 & 0 & 1 & 0 & 0 & 1 & 0 & 1 & 0 \\
0 & 0 & 0 & 0 & 0 & 1 & 0 & 0 & 0 & 1 & 0 & 1 & 0 & 1 & 0 & 0 & 0 & 0 & 0 &
1+u & 1 & u & 1 & u & u & 1+u & 1 & 0 & 0 & 1 & 0 & 0 & 1 & 0 & 1 \\
0 & 1 & 1 & 0 & 0 & 1 & 1 & 1 & 1 & 1 & 1 & 0 & 0 & 0 & u & u & 0 & u & u & 0
& u & 0 & u & u & 0 & u & u & 0 & 0 & 0 & 0 & 0 & 0 & 0 & 0 \\
0 & 0 & 1 & 1 & 0 & 0 & 1 & 1 & 1 & 1 & 1 & 1 & 0 & 0 & 0 & u & u & 0 & u & u
& 0 & u & 0 & u & u & 0 & u & u & 0 & 0 & 0 & 0 & 0 & 0 & 0 \\
0 & 0 & 0 & 1 & 1 & 0 & 0 & 1 & 1 & 1 & 1 & 1 & 1 & 0 & 0 & 0 & u & u & 0 & u
& u & 0 & u & 0 & u & u & 0 & u & u & 0 & 0 & 0 & 0 & 0 & 0 \\
0 & 0 & 0 & 0 & 1 & 1 & 0 & 0 & 1 & 1 & 1 & 1 & 1 & 1 & 0 & 0 & 0 & u & u & 0
& u & u & 0 & u & 0 & u & u & 0 & u & u & 0 & 0 & 0 & 0 & 0 \\
1 & 0 & 0 & 0 & 0 & 1 & 1 & 0 & 0 & 1 & 1 & 1 & 1 & 1 & 0 & 0 & 0 & 0 & u & u
& 0 & u & u & 0 & u & 0 & u & u & 0 & u & u & 0 & 0 & 0 & 0 \\
1 & 1 & 0 & 0 & 0 & 0 & 1 & 1 & 0 & 0 & 1 & 1 & 1 & 1 & 0 & 0 & 0 & 0 & 0 & u
& u & 0 & u & u & 0 & u & 0 & u & u & 0 & u & u & 0 & 0 & 0 \\
1 & 1 & 1 & 0 & 0 & 0 & 0 & 1 & 1 & 0 & 0 & 1 & 1 & 1 & 0 & 0 & 0 & 0 & 0 & 0
& u & u & 0 & u & u & 0 & u & 0 & u & u & 0 & u & u & 0 & 0 \\
1 & 1 & 1 & 1 & 0 & 0 & 0 & 0 & 1 & 1 & 0 & 0 & 1 & 1 & 0 & 0 & 0 & 0 & 0 & 0
& 0 & u & u & 0 & u & u & 0 & u & 0 & u & u & 0 & u & u & 0 \\
1 & 1 & 1 & 1 & 1 & 0 & 0 & 0 & 0 & 1 & 1 & 0 & 0 & 1 & 0 & 0 & 0 & 0 & 0 & 0
& 0 & 0 & u & u & 0 & u & u & 0 & u & 0 & u & u & 0 & u & u%
\end{smallmatrix}
\right)
\end{eqnarray*}

\begin{equation*}
H=\left(
\begin{smallmatrix}
0 & 0 & 0 & 0 & 0 & 0 & 1 & 1 & 1 & 1 & 0 & 0 & 1 & 1 & 0 & 0 & 0 & 0 & 0 & 0
& 0 & 0 & 0 & 0 & 0 & 0 & 0 & 0 & 0 & 0 & 0 & 0 & 0 & 0 & 0 \\
1 & 0 & 0 & 0 & 0 & 0 & 0 & 1 & 1 & 1 & 1 & 0 & 0 & 1 & 0 & 0 & 0 & 0 & 0 & 0
& 0 & 0 & 0 & 0 & 0 & 0 & 0 & 0 & 0 & 0 & 0 & 0 & 0 & 0 & 0 \\
1 & 1 & 0 & 0 & 0 & 0 & 0 & 0 & 1 & 1 & 1 & 1 & 0 & 0 & 0 & 0 & 0 & 0 & 0 & 0
& 0 & 0 & 0 & 0 & 0 & 0 & 0 & 0 & 0 & 0 & 0 & 0 & 0 & 0 & 0 \\
0 & 1 & 1 & 0 & 0 & 0 & 0 & 0 & 0 & 1 & 1 & 1 & 1 & 0 & 0 & 0 & 0 & 0 & 0 & 0
& 0 & 0 & 0 & 0 & 0 & 0 & 0 & 0 & 0 & 0 & 0 & 0 & 0 & 0 & 0 \\
0 & 0 & 1 & 1 & 0 & 0 & 0 & 0 & 0 & 0 & 1 & 1 & 1 & 1 & 0 & 0 & 0 & 0 & 0 & 0
& 0 & 0 & 0 & 0 & 0 & 0 & 0 & 0 & 0 & 0 & 0 & 0 & 0 & 0 & 0 \\
1 & 0 & 0 & 1 & 1 & 0 & 0 & 0 & 0 & 0 & 0 & 1 & 1 & 1 & 0 & 0 & 0 & 0 & 0 & 0
& 0 & 0 & 0 & 0 & 0 & 0 & 0 & 0 & 0 & 0 & 0 & 0 & 0 & 0 & 0 \\
1 & 1 & 0 & 0 & 1 & 1 & 0 & 0 & 0 & 0 & 0 & 0 & 1 & 1 & 0 & 0 & 0 & 0 & 0 & 0
& 0 & 0 & 0 & 0 & 0 & 0 & 0 & 0 & 0 & 0 & 0 & 0 & 0 & 0 & 0 \\
0 & 0 & 0 & 0 & 1 & 1 & 1 & 0 & 0 & 1 & 0 & 0 & 0 & 0 & 0 & 0 & 0 & 0 & 0 & 0
& 0 & 0 & 1 & 0 & 0 & 0 & 0 & 0 & 1+u & 0 & u & 1 & u & u & 1+u \\
0 & 0 & 0 & 0 & 0 & 1 & 1 & 1 & 0 & 0 & 1 & 0 & 0 & 0 & 1+u & 0 & 0 & 0 & 0
& 0 & 0 & 0 & 0 & 1 & 0 & 0 & 0 & 0 & 0 & 1+u & 0 & u & 1 & u & u \\
0 & 0 & 0 & 0 & 0 & 0 & 1 & 1 & 1 & 0 & 0 & 1 & 0 & 0 & u & 1+u & 0 & 0 & 0
& 0 & 0 & 0 & 0 & 0 & 1 & 0 & 0 & 0 & 0 & 0 & 1+u & 0 & u & 1 & u \\
0 & 0 & 0 & 0 & 0 & 0 & 0 & 1 & 1 & 1 & 0 & 0 & 1 & 0 & u & u & 1+u & 0 & 0
& 0 & 0 & 0 & 0 & 0 & 0 & 1 & 0 & 0 & 0 & 0 & 0 & 1+u & 0 & u & 1 \\
0 & 0 & 0 & 0 & 0 & 0 & 0 & 0 & 1 & 1 & 1 & 0 & 0 & 1 & 1 & u & u & 1+u & 0
& 0 & 0 & 0 & 0 & 0 & 0 & 0 & 1 & 0 & 0 & 0 & 0 & 0 & 1+u & 0 & u \\
1 & 0 & 0 & 0 & 0 & 0 & 0 & 0 & 0 & 1 & 1 & 1 & 0 & 0 & u & 1 & u & u & 1+u
& 0 & 0 & 0 & 0 & 0 & 0 & 0 & 0 & 1 & 0 & 0 & 0 & 0 & 0 & 1+u & 0 \\
0 & 1 & 0 & 0 & 0 & 0 & 0 & 0 & 0 & 0 & 1 & 1 & 1 & 0 & 0 & u & 1 & u & u &
1+u & 0 & 0 & 0 & 0 & 0 & 0 & 0 & 0 & 1 & 0 & 0 & 0 & 0 & 0 & 1+u \\
0 & 0 & 1 & 0 & 0 & 0 & 0 & 0 & 0 & 0 & 0 & 1 & 1 & 1 & 1+u & 0 & u & 1 & u
& u & 1+u & 0 & 0 & 0 & 0 & 0 & 0 & 0 & 0 & 1 & 0 & 0 & 0 & 0 & 0 \\
1 & 0 & 0 & 1 & 0 & 0 & 0 & 0 & 0 & 0 & 0 & 0 & 1 & 1 & 0 & 1+u & 0 & u & 1
& u & u & 1+u & 0 & 0 & 0 & 0 & 0 & 0 & 0 & 0 & 1 & 0 & 0 & 0 & 0 \\
0 & 0 & 1 & 1 & 0 & 0 & 1 & 1 & 1 & 1 & 1 & 1 & 0 & 0 & 0 & 0 & 0 & 0 & u & 0
& u & 0 & u & u & 0 & 0 & u & u & u & 0 & u & u & u & u & 0 \\
0 & 0 & 0 & 1 & 1 & 0 & 0 & 1 & 1 & 1 & 1 & 1 & 1 & 0 & 0 & 0 & 0 & 0 & 0 & u
& 0 & u & 0 & u & u & 0 & 0 & u & u & u & 0 & u & u & u & u \\
0 & 0 & 0 & 0 & 1 & 1 & 0 & 0 & 1 & 1 & 1 & 1 & 1 & 1 & u & 0 & 0 & 0 & 0 & 0
& u & 0 & u & 0 & u & u & 0 & 0 & u & u & u & 0 & u & u & u \\
1 & 0 & 0 & 0 & 0 & 1 & 1 & 0 & 0 & 1 & 1 & 1 & 1 & 1 & u & u & 0 & 0 & 0 & 0
& 0 & u & 0 & u & 0 & u & u & 0 & 0 & u & u & u & 0 & u & u \\
1 & 1 & 0 & 0 & 0 & 0 & 1 & 1 & 0 & 0 & 1 & 1 & 1 & 1 & u & u & u & 0 & 0 & 0
& 0 & 0 & u & 0 & u & 0 & u & u & 0 & 0 & u & u & u & 0 & u \\
1 & 1 & 1 & 0 & 0 & 0 & 0 & 1 & 1 & 0 & 0 & 1 & 1 & 1 & u & u & u & u & 0 & 0
& 0 & 0 & 0 & u & 0 & u & 0 & u & u & 0 & 0 & u & u & u & 0%
\end{smallmatrix}
\right).
\end{equation*}

Hence, ${\mathcal{C}}$ is a self-orthogonal code and $\Phi({\mathcal{C}})$
is a $[56,25,6]$ binary code and $\Phi({\mathcal{C}}^{\perp})$ is a $%
[56,31,6]$ binary code. We can conclude that there exists a $[[56,6,6]]_2$
stabilizer code.
\end{example}

\begin{example}
Let $\mathcal{C}=\left\langle(f(x),0),(l(x),g(x)+ua(x))\right\rangle $ be a
cyclic code in $\mathbb{Z}_{2}[x]/\langle x^{6}-1\rangle\times R[x]/\langle
x^{5}-1\rangle$ where
\begin{eqnarray*}
f(x) &=&x^{6}-1, \\
g(x) &=&x^{5}-1, \\
a(x) &=&1 + x + x^2 + x^3 + x^4, \\
l(x) &=&1+x+x^2+x^3+x^4+x^5.
\end{eqnarray*}

Therefore, we can write the generator and parity-check matrices for ${%
\mathcal{C}}$ using the same way above. Then we have,
\begin{equation*}
G=\left(
\begin{array}{ccccccccccc}
1 & 1 & 1 & 1 & 1 & 1 & u & u & u & u & u%
\end{array}
\right)
\end{equation*}
\begin{equation*}
H=\left(
\begin{array}{ccccccccccc}
0 & 0 & 0 & 0 & 1 & 1 & 0 & 0 & 0 & 0 & 0 \\
1 & 0 & 0 & 0 & 0 & 1 & 0 & 0 & 0 & 0 & 0 \\
1 & 1 & 0 & 0 & 0 & 0 & 0 & 0 & 0 & 0 & 0 \\
0 & 1 & 1 & 0 & 0 & 0 & 0 & 0 & 0 & 0 & 0 \\
0 & 0 & 1 & 1 & 0 & 0 & 0 & 0 & 0 & 0 & 0 \\
0 & 0 & 0 & 0 & 0 & 1 & 1+u & 0 & 0 & 0 & 0 \\
1 & 0 & 0 & 0 & 0 & 0 & 0 & 1+u & 0 & 0 & 0 \\
0 & 1 & 0 & 0 & 0 & 0 & 0 & 0 & 1+u & 0 & 0 \\
0 & 0 & 1 & 0 & 0 & 0 & 0 & 0 & 0 & 1+u & 0 \\
0 & 0 & 0 & 1 & 0 & 0 & 0 & 0 & 0 & 0 & 1+u%
\end{array}
\right)
\end{equation*}
Here, $\Phi({\mathcal{C}})$ is a binary $[16,1,16]$ code where $\Phi({%
\mathcal{C}}^{\perp})$ is a $[16,15,2]$ binary code. Then we can conclude
that there exits [[16,14,2]] quantum code which has optimal parameters \cite%
{8}.
\end{example}

\begin{example}
Let $\mathcal{C}$ be a cyclic code in $\mathbb{Z}_{2}[x]/\langle
x^{8}-1\rangle\times R[x]/\langle x^{5}-1\rangle$ with the following
generator polynomials.
\begin{eqnarray*}
f(x) &=&1+x+x^2+x^3+x^4+x^5+x^6+x^7, \\
g(x) &=&x^{5}-1, \\
a(x) &=&1 + x + x^2 + x^3 + x^4, \\
l(x) &=&1 + x^2 + x^4 + x^6.
\end{eqnarray*}

Then, the generator and parity-check matrices of ${\mathcal{C}}$ are
\begin{equation*}
G=\left(
\begin{array}{ccccccccccccc}
1 & 1 & 1 & 1 & 1 & 1 & 1 & 1 & 0 & 0 & 0 & 0 & 0 \\
1 & 0 & 1 & 0 & 1 & 0 & 1 & 0 & u & u & u & u & u%
\end{array}
\right)
\end{equation*}

\begin{equation*}
H=\left(
\begin{array}{ccccccccccccc}
0 & 0 & 0 & 0 & 0 & 1 & 0 & 1 & 0 & 0 & 0 & 0 & 0 \\
1 & 0 & 0 & 0 & 0 & 0 & 1 & 0 & 0 & 0 & 0 & 0 & 0 \\
0 & 1 & 0 & 0 & 0 & 0 & 0 & 1 & 0 & 0 & 0 & 0 & 0 \\
1 & 0 & 1 & 0 & 0 & 0 & 0 & 0 & 0 & 0 & 0 & 0 & 0 \\
0 & 1 & 0 & 1 & 0 & 0 & 0 & 0 & 0 & 0 & 0 & 0 & 0 \\
0 & 0 & 1 & 0 & 1 & 0 & 0 & 0 & 0 & 0 & 0 & 0 & 0 \\
0 & 0 & 0 & 0 & 0 & 1 & 1 & 0 & 0 & 0 & 0 & 0 & 1+u \\
0 & 0 & 0 & 0 & 0 & 0 & 1 & 1 & 1+u & 0 & 0 & 0 & 0 \\
1 & 0 & 0 & 0 & 0 & 0 & 0 & 1 & 0 & 1+u & 0 & 0 & 0 \\
1 & 1 & 0 & 0 & 0 & 0 & 0 & 0 & 0 & 0 & 1+u & 0 & 0 \\
0 & 1 & 1 & 0 & 0 & 0 & 0 & 0 & 0 & 0 & 0 & 1+u & 0%
\end{array}
\right)
\end{equation*}
where, $\Phi({\mathcal{C}})$ is a $[18,2,8]_{2}$ code and $\Phi({\mathcal{C}}%
^{\perp})$ is a $[18,16,2]_{2}$ code. Therefore, we have [[18,14,2]] quantum
code which is optimal.
\end{example}

\section{Conclusion}

In this work, we study self-dual cyclic codes over ${\mathbb{Z}_{2}}^{\alpha}\times R^{\beta}$ where $R=\{0,1,u,1+u\}$ with $u^{2}=0$. We obtain
quantum codes using self-orthogonal ${\mathbb{Z}_{2}\mathbb{Z}_{2}[u]}$%
-cyclic codes via CSS quantum code construction. We also give some examples
of both self-dual cyclic and quantum codes. Some of the examples of
quantum codes we provided are optimal which makes this class of codes an
interesting class to be studied further.

\end{document}